%% file: paper.tex
\documentclass[11pt]{article}
\usepackage[margin=1in]{geometry}
\usepackage{times}
\usepackage{wrapfig}
\usepackage{citesort}
\usepackage{epsfig}
\usepackage{graphicx}
\usepackage{color}
\usepackage{fullpage}
\usepackage{amssymb}
\usepackage{amsmath}
\usepackage{theorem}
\usepackage[ruled, vlined, linesnumbered]{algorithm2e}
\usepackage{color}

\newcommand\T{\rule{0pt}{2.6ex}}       

\newcommand{\eps}{\epsilon}
\newcommand{\E}{\mathbf{E}}
\renewcommand{\Pr}{\mathbf{Pr}}
\newcommand{\abs}[1]{\left| #1 \right|}

\newcommand{\sAMS}{\overset{\tt{AMS}}{\sim}}
\newcommand{\CEN}{{\em CountEachSimple}}
\newcommand{\CE}{{\em CountEach}}
\newcommand{\CA}{{\em CountAll}}

\newcommand{\CM}{{\em CountMin}}
\newcommand{\EstR}{{\em TrackR}}
\newcommand{\EstP}{{\em TrackProb}}
\newcommand{\EstH}{{\em TrackEntropy}}
\newcommand{\Est}{{\sf Est}}

\newtheorem{theorem}{Theorem}

\newtheorem{lemma}{Lemma}

\newtheorem{corollary}{Corollary}

{\theorembodyfont{\rmfamily}
\newtheorem{remark}{Remark}}
{\theorembodyfont{\rmfamily}
\newtheorem{definition}{Definition}}

\newcommand{\qinhides}[1]{}

\newenvironment{proof}{\trivlist\item[]\emph{Proof:}}%
{\unskip\nobreak\hskip 1em plus 1fil\nobreak$\Box$
\parfillskip=0pt%
\endtrivlist}

\begin{document}

\title{Improved Algorithms for Distributed Entropy Monitoring\thanks{Work supported in part by NSF CCF-1525024, IIS-1633215, and IU’s
Office of the Vice Provost for Research through the Faculty
Research Support Program.}}

\author{Jiecao Chen\\Indiana University Bloomington\\
              jiecchen@umail.iu.edu 
\and Qin Zhang\thanks{Corresponding author.}\\Indiana University Bloomington\\
              qzhangcs@indiana.edu}
  
\date{}

\maketitle

\begin{abstract}
\input{abstract}
\end{abstract}

\input{intro}


\input{ams}

\input{entropy}



\input{conclude}

\bibliographystyle{abbrv}
\bibliography{paper,rest,monitor,mypub}


\appendix

\input{appendix}

\end{document}

%% file: abstract.tex
Modern data management systems often need to deal with massive, dynamic and inherently distributed data sources. We collect the data using a distributed network, and at the same time try to maintain a global view
of the data at a central coordinator using a minimal amount of communication. Such applications have been captured by the distributed monitoring model which has attracted a lot of attention in recent years.  In this paper we investigate the monitoring of the entropy functions, which are very useful in network monitoring applications such as detecting distributed denial-of-service attacks. 
Our results improve the previous best results by Arackaparambil et al.~\cite{ABC09}.   Our technical contribution also includes implementing the celebrated AMS sampling method (by Alon et al.~\cite{AMS99}) in the distributed monitoring model, which could be of independent interest.

%% file: intro.tex
\section{Introduction}

Modern data management systems often need to deal with massive, dynamic, and inherently distributed data sources, such as packets passing through the IP backbone network, loads of machines in content delivery service systems, data collected by large-scale environmental sensor networks, etc. 
One of the primary goals is to detect important and/or abnormal events that have happened in  networks and systems in a timely manner, while incurring a minimal amount of communication overhead. These applications led to the study of   {\em (continuous) distributed monitoring}, which has attracted a lot of attention in the database and network communities in the past decade \cite{DR01,OJW03,BO03,JHRW04,GGR04,MSDO05,CGMR05,CG05,CMZ06,KCR06,SSK06,CMZ07,SSK08}. The model was then formalized by Cormode, Muthukrishnan and Yi~\cite{CMY08} in 2008, and since then considerable work has been done in theory, including tracking heavy hitters and quantiles~\cite{YZ09,HYZ12}, entropy~\cite{ABC09}, frequency moments~\cite{CMY08,WZ12}, and performing random sampling~\cite{CMYZ12,TW11}. Some of these problems have also been studied in the sliding window settings~\cite{CMYZ12,CLLT12,CY12}.  We note that a closely related model, called the {\em distributed streams} model, has been proposed and studied earlier by Gibbons and Tirthapura~\cite{GT01,GT02}. The distributed streams model focuses on one-shot computation, and is thus slightly different from the (continuous) distributed monitoring model considered in this paper.


In this paper we focus on monitoring the entropy functions.  Entropy is one of the most important functions in distributed monitoring due to its effectiveness in detecting distributed denial-of-service attacks (the empirical entropy of IP addresses passing through a router may exhibit a noticeable change under an attack), clustering to reveal interesting patterns and performing traffic classifications~\cite{XZB05,LCD05,LSOXZ06} (different values of empirical entropy correspond to different traffic patterns), which are central tasks in network monitoring. 
Previously the entropy problem has also been studied extensively in the data stream model~\cite{LSOXZ06,GMV06,CBM06,CCM10,HNO08}.
 Arackaparambil, Brody and Chakrabarti~\cite{ABC09} studied the Shannon and Tsallis entropy functions in the distributed {\em threshold} monitoring setting, where one needs to tell whether the entropy of the joint data stream $\ge \tau$ or $ \le \tau(1 - \eps)$ for a fixed threshold $\tau$ at any time step. 
They obtained algorithms using $\tilde{O}(k/(\eps^3\tau^3))$\footnote{$\tilde{O}(\cdot)$ ignores all polylog terms; see Table \ref{tab:notation} for details.} bits of communication where $k$ is the number of sites. Note that this bound can be arbitrarily large when $\tau$ is arbitrarily small. In this paper we design algorithms that can track these two entropy functions continuously using $\tilde{O}(k/\eps^2 + \sqrt{k}/\eps^3)$ bits of communication.  Our results work for all $\tau \ge 0$ simultaneously, and have improved upon the  results in \cite{ABC09} by at least a factor of $\min\{\sqrt{k}, 1/\eps\}$ (ignoring logarithmic factors), even for a constant $\tau$. Another advantage of our algorithms is that they can be easily extended to the sliding window cases (the approximation error needs to be an additive $\eps$ when $\tau$ is small).  

As a key component of our algorithms, we show how to implement the AMS sampling, a method initially proposed by Alon, Matias and Szegedy~\cite{AMS99} in the data stream model, in distributed monitoring.  Similar as that in the data stream model, this sampling procedure can be used to track general (single-valued) functions in the distributed monitoring model, and should be of independent interest.

We note that the techniques used in our algorithms are very different from that  in \cite{ABC09}.  The algorithm in \cite{ABC09} monitors the changes of the values of the entropy function over time, and argue that it has to ``consume'' many  items in the input streams for the function value to change by a factor of $\Theta(\eps \tau)$ ($\eps$ is the approximation error and $\tau$ is the threshold). We instead adapt the AMS sampling framework to the distributed monitoring model, as we will explain shortly in the technique overview.  We also note that using the AMS sampling for monitoring the entropy function has already been conducted  in the data stream model \cite{CCM10}, but due to the model differences it takes quite some non-trivial efforts to port this idea to the distributed monitoring model.

In the rest of this section, we will first define the distributed monitoring model, and then give an overview of the techniques used in our algorithms.
\vspace{3mm}

\noindent{\bf The Distributed Monitoring Model.}
In the distributed monitoring model, we have $k$ sites $S_1, \ldots, S_k$ and one central coordinator. Each site observes a stream $A_i$ of items  over time. Let $A = (a_1, \ldots, a_m) \in [n]^m$ be the joint global stream (that is, the concatenation of $A_i$'s with items ordered by their arrival times). Each $a_\ell$ is observed by exactly one of the $k$ sites at time $t_\ell$, where $t_1 < t_2 < \ldots < t_m$. Let $A(t)$ be the set of items received by the system until time $t$, and let $f$ be the function we would like to track. The coordinator is required to report $f(A(t))$ at {\em any} time step $t$. There is a two-way communication channel between each site and the coordinator. Our primary goal is to minimize the total bits of communication between sites and the coordinator in the whole process, since the communication cost directly links to the network bandwidth usage and energy consumption.\footnote{It is well-known that in sensor networks, communication is by far the biggest battery drain.~\cite{MFHH03}}
We also want to minimize the space usage and processing time per item at each site. Generally speaking, we want the total communication,  space usage and processing time per item to be {\em sublinear} in terms of input size $m$ and the item universe size $n$.

We also consider two standard {\em sliding window} settings, namely, the {\em sequence-based} sliding window and the {\em time-based} sliding window. In the sequence-based case, at any time step $t_{now}$ the coordinator is required to report $$f(A^w(t_{now})) = f(a_{L - w + 1}, \ldots, a_L),$$ where $w$ is the length of the sliding window and $L = \max\{\ell\ |\ t_\ell \le t_{now}\}$. In other words, the coordinator needs to  maintain the value of the function defined on the most recent $w$ items continuously. In the time-based case the coordinator needs to maintain the function on the items that are received in the last $t$ time steps, that is, on $A^t = A(t_{now}) \backslash A({t_{now} - t})$. To differentiate we call the full stream case the {\em infinite window}.

\vspace{3mm}

\begin{table}[t]
\begin{tabular}{|l|l|l|l|l|l|l|l| }
\hline
function & type & window & approx & (total) comm. & space (site) & time (site) & ref. \\
\hline
Shannon & threshold & infinite & multi. & $\tilde{O}({k}/{(\eps^3\tau^3)})$   & $\tilde{O}(\eps^{-2})$  & -- & \cite{ABC09} \\
\hline
Tsallis & threshold & infinite & multi. & $\tilde{O}({k}/{(\eps^3\tau^3)})$ & $\tilde{O}(\eps^{-2})$ & -- & \cite{ABC09}  \\
\hline
Shannon & continuous & infinite &  multi. & $\tilde{O}(k/\eps^2 + \sqrt{k}/\eps^3)$ & $\tilde{O}(\eps^{-2})$ &  $\tilde{O}(\eps^{-2})$ & new \\
\hline
Shannon & continuous & sequence & mixed & $\tilde{O}(k/\eps^2 + \sqrt{k}/\eps^3)$ & $\tilde{O}(\eps^{-2})$ & $\tilde{O}(\eps^{-2})$ & new \\
\hline
Shannon & continuous & time & mixed & $\tilde{O}(k/\eps^2 + \sqrt{k}/\eps^3)$ & $\tilde{O}(\eps^{-2})$ & $\tilde{O}(\eps^{-2})$ & new \\
\hline
Tsallis & continuous & infinite & multi. & $\tilde{O}(k/\eps^2 + \sqrt{k}/\eps^3)$ & $\tilde{O}(\eps^{-2})$  & $\tilde{O}(\eps^{-2})$ & new  \\
\hline
Tsallis & continuous & sequence & mixed & $\tilde{O}(k/\eps^2 + \sqrt{k}/\eps^3)$ & $\tilde{O}(\eps^{-2})$  & $\tilde{O}(\eps^{-2})$ & new \\
\hline
Tsallis & continuous & time & mixed & $\tilde{O}(k/\eps^2 + \sqrt{k}/\eps^3)$ & $\tilde{O}(\eps^{-2})$  & $\tilde{O}(\eps^{-2})$ & new \\
\hline
\end{tabular}
\caption{Summary of results.
{\em Sequence} and {\em time} denote sequence-based sliding window and time-based sliding window respectively. The correctness guarantees for the sliding window cases are different from the infinite-window case: {\em multi.}\ stands for $(1+\eps, \delta)$-approximation; and {\em mixed} stands for $(1+\eps, \delta)$-approximation when the entropy is larger than $1$, and $(\eps, \delta)$-approximation when the entropy is at most $1$. In the threshold monitoring model, $\tau$ is the threshold value.}  
\label{tab:results}
\end{table}

\noindent{\bf Our results.}  In this paper we study the following two entropy functions.
\begin{itemize}
\item {\em Shannon entropy (also known as empirical entropy).}  For an input sequence $A$ of length $m$, the Shannon entropy is defined as $H(A) = \sum_{i\in[n]}\frac{m_i}{m}\log\frac{m}{m_i}$ where $m_i$ is the frequency of the $i$th element. 

\item {\em Tsallis entropy.} The $q$-th $(q > 1)$ Tsallis entropy is defined as $T_q(A) =\frac{1 - \sum_{i \in [n]} (\frac{m_i}{m})^q}{q-1}$. It is known that when $q\rightarrow 1$, the $q$-th Tsallis entropy converges to  the Shannon entropy.  
\end{itemize}
We say $\hat{Q}$ is $(1+\eps, \delta)$-approximation of $Q$ iff $\Pr[|Q - \hat{Q}| > \epsilon Q ] \leq \delta$, and $\hat{Q}$ is $(\eps, \delta)$-approximation of $Q$ iff $\Pr[|Q - \hat{Q}| > \epsilon] \leq \delta$.
Our results are summarized in Table \ref{tab:results}. 
Note that $\log (1/\delta)$ factors are absorbed in the $\tilde{O}(\cdot)$ notation.
\vspace{3mm}

\noindent{\bf Technique Overview.}
We first recall the AMS sampling method in the data stream model. Let $A = \{a_1, \ldots, a_m\} \in [n]^m$ be the stream. The AMS sampling consists of three steps: (1)  pick $J \in [m]$ uniformly at random; (2) let $R = \abs{\{j: a_j=a_J, J \leq j \leq m\}}$ be the frequency of the element $a_J$ in the rest of the stream (call it $a_J$'s {\em tail frequency}); and (3) set $X =  f(R) - f(R-1)$. In the rest of the paper we will use  $(a_J, R) \sAMS A$ to denote the first two steps. Given $(m_1, m_2, \ldots, m_n)$ as the frequency vector of the data stream $A$, letting $\bar{f}(A) = \frac{1}{m} \sum_{i \in [n]} f(m_i)$,  it has been shown  that $\E[X] = \bar{f}(A)$ \cite{AMS99}.\footnote{To see this, note that 
$\E[X] = \sum_{j \in [m]} \E[X | a_J = j] \Pr[a_J = j] = \sum_{j \in [m]}  \left(\E[f(R) - f(R-1) | a_J = j] \cdot m_j / m\right)$, and $\E[f(R) - f(R-1) | a_J = j] = \sum_{k \in [m_j]} \frac{f(k) - f(k-1)}{m_j} = \frac{f(m_j)}{m_j}$.}
By the standard repeat-average technique (i.e. run multiple independent copies in parallel and take the average of the outcomes), we can use sufficient (possibly polynomial in $n$, but for entropy this is $\tilde{O}(1/\eps^2)$) i.i.d.\ samples of $X$ to get a $(1+\eps)$-approximation of $\bar{f}(A)$. 

A key component of our algorithms is to implement $(a_J, R) \sAMS A$ in distributed monitoring. Sampling $a_J$ can be done using a random sampling algorithm by Cormode et al.~\cite{CMYZ12}. Counting $R$ seems to be easy; however, in distributed monitoring $\Omega(m)$ bits of communication are needed if we want to keep track of $R$ exactly at any time step. One of our key observations is that a $(1+\eps)$-approximation of $R$ should be enough for a big class of functions, and we can use any existing counting algorithms (e.g., the one by Huang et al.~\cite{HYZ12}) to maintain such an approximation of $R$. Another subtlety is that the sample $a_J$ will change over time, and for every change we have to restart the counting process. Fortunately, we manage to bound the number of updates of $a_J$ by $O(\log m)$.

To apply the AMS sampling approach to the Shannon entropy functions efficiently, we need to tweak the framework a bit. The main reason is that the AMS sampling method works poorly on an input that has a very frequent element, or equivalently, when the entropy of the input is close to $0$. At a high level, our algorithms adapt the techniques developed by Chakrabarti et al.\ \cite{CCM10} for computing entropy in the data stream model where they track $p_{\max}$ as the empirical probability of the most frequent element $i_{\max}$, and approximate $\bar{f}(A)$ by 
$$(1-p_{\max}) \bar{f}(A\backslash i_{\max}) +  p_{\max}\log \frac{1}{p_{\max}},$$
where for a universe element $a$, $A\backslash a$ is the substream of $A$ obtained by removing all occurrences of $a$ in $A$ while keeping the orders of the rest of the items.
But due to inherent differences between (single) data stream and distributed monitoring, quite a few specific implementations need to be reinvestigated, and the analysis is also different since in distributed monitoring we primarily care about communication instead of space. For example, it is much more complicated to track $(1 - p_{\max})$ up to a $(1+\eps)$ approximation in distributed monitoring than in the data stream model , for which we need to assemble a set of  tools developed in previous work \cite{cormode05,YZ09,HYZ12}. 
\vspace{3mm}

\noindent{\bf Notations and conventions.}
We summarize the main notations in this paper in Table~\ref{tab:notation}.  We differentiate 
{\em item} and {\em element};  we use {\em item} to denote a token in the stream $A$, and {\em element} to denote an element from the universe $[n]$.  We refer to $\eps$ as {\em approximation error}, and $\delta$ as {\em failure probability}.

\begin{table}[t]
\centering
\begin{tabular}{|l|l|}
\hline
$k$ & number of sites  \\
\hline
$[n]$ & $[n] = \{1, 2, \ldots,  n\}$, the universe \\
\hline
$-i$  & $-i = [n]\backslash i = \{1, \ldots, i-1, i+1, \ldots, n \}$ \\
\hline
$s \in_R S$ & the process of sampling $s$ from set $S$ uniformly at random \\
\hline
$\log x, \ln x$ & $\log x = \log_2 x$, $\ln x = \log_e x$ \\
\hline
$A$ &  $A = (a_1, \ldots, a_m) \in [n]^m$ is a sequence of items \\
\hline
$A \backslash z$ &  subsequence obtained by deleting all the occurrences of $z$ from $A$ \\
\hline
$m_i$ & $m_i = \abs{\{j : a_j = i \}}$ is the frequency of element $i$ in $A$\\
\hline
$p_i$ & $p_i = \frac{m_i}{m}$, the empirical probability of $i$ \\
\hline
$\vec{p}$ & $\vec{p} = (p_1, p_2, \ldots, p_n)$ \\
\hline
$H(A) \equiv H(\vec{p})$ & $H(A) \equiv H(\vec{p}) = \sum_{i\in[n]} p_i\log p_i^{-1}$ is the Shannon entropy of $A$ \\
\hline
$m_{-i}$ & $m_{-i}= \sum_{j\in [n]\backslash i} m_j$ \\
\hline
$\bar{f}(A)$ & $\bar{f}(A) = \frac{1}{\abs{A}} \sum_{i\in[n]}f(m_i)$ \T \\
\hline
$H(A), f_m$ & $f_m(x) = x \log \frac{m}{x}$ and $H(A)\equiv \bar{f_m}(A)$ \T \\
\hline
$(1+\epsilon, \delta)$-approx. & $\hat{Q}$ is $(1+\epsilon, \delta)$-approximation of $Q$ iff $\Pr[|Q - \hat{Q}| > \epsilon Q ] \leq \delta$ \T \\
\hline
$(1 + \eps)$-approx. & simplified notation for $(1 + \eps, 0)$-approximation \\
\hline
$(\eps, \delta)$-approx. & $\hat{Q}$ is $(\eps, \delta)$-approximation of $Q$ iff $\Pr[ |Q - \hat{Q}| > \eps] \leq \delta$ \T\\
\hline
$\eps$-approx. & simplified notation for $(\eps, 0)$-approximation\\
\hline
$\tilde{O}(\cdot)$ & $\tilde{O}$ suppresses poly$(\log \frac{1}{\eps}, \log \frac{1}{\delta}, \log n, \log m)$ \T\\
\hline
$\text{Est}(f, R, \kappa)$ & defined in Section \ref{sec:AMS:pre} \\
\hline
$\lambda_{f, \mathcal{A}}$ &  see Definition \ref{def:trackable}\\
\hline
$\lambda$ &  when $f$ and $\mathcal{A}$ are clear from context, $\lambda$ is short for $\lambda_{f, \mathcal{A}}$ \\  
\hline
\end{tabular}
\caption{List of notations}
\label{tab:notation} 
\end{table}

\vspace{3mm}

\noindent{\bf Roadmap.}
In Section~\ref{sec:AMS} we show how to implement the AMS sampling in the distributed monitoring model, which will be used in our entropy monitoring algorithms.  We present our improved algorithms for monitoring the Shannon entropy function and the Tsallis entropy function in Section~\ref{sec:Shannon} and Section~\ref{sec:Tsallis}, respectively. We then conclude the paper in Section~\ref{sec:conclude}.

%% file: ams.tex
\section{AMS Sampling in Distributed Monitoring}
\label{sec:AMS}

In this section we extend the AMS sampling algorithm to the distributed monitoring model.  We choose to present this implementation in a general form so that it can be used for tracking both the Shannon entropy and the Tsallis entropy. We will discuss both the infinite window case and the sliding window cases. 
\vspace{3mm}

\noindent{\bf Roadmap.} We will start by introducing some tools from previous work, and then give the algorithms for the infinite window case, followed by the analysis.  We then discuss the sliding window cases.

\subsection{Preliminaries}
\label{sec:AMS:pre}
Recall the AMS sampling framework sketched in the introduction. Define $\text{\Est}(f, R, \kappa) = \frac{1}{\kappa} \sum_{i \in [\kappa]} X_i$, where $\{X_1, \ldots, X_\kappa\}$ are i.i.d.\ sampled from the distribution of $X = f(R) - f(R-1)$. The following lemma shows that for a sufficiently large $\kappa$, $\text{\Est}(f, R, \kappa)$ is a good estimator of $\E[X]$.

\begin{lemma}[\cite{CCM10}]
  \label{lem:c}
  Let $a\geq 0, b > 0$ such that $-a \leq X \leq b$, and
  \begin{equation}
    \label{eq:inq:c}
    \kappa \geq \frac{3(1 + a/\E[X])^2\eps^{-2}\ln(2\delta^{-1})(a + b)}{(a + \E[X])}.
  \end{equation}
If $\E[X] > 0$, then $\text{\Est}(f, R, \kappa)$ gives a $(1 + \eps, \delta)$-approximation to $\E[X] = \bar{f}(A)$.
\end{lemma}

We will also make use of the following tools from the previous work in distributed monitoring.
\vspace{3mm}

\noindent{\bf \CEN.}
A simple (folklore) algorithm for counting the frequency of a given element in distribution monitoring is the following: Each site $S_i$ maintains a local counter $ct_i$, initiated to be $1$. Every time $ct_i$ increases by a factor of $(1+\eps)$, $S_i$ sends a message (say, a signal bit) to the coordinator. It is easy to see that the coordinator can always maintain a $(1+\eps)$-approximation of $\sum_i ct_i$, which is the frequency of the element. The total communication cost can be bounded by $O(k \cdot \log_{1+\eps} m) = O(k/\eps \cdot \log m)$ bits. The space used at each site is $O(\log m)$ bits and the processing time per item is $O(1)$. We denote this algorithm by \CEN$(e, \eps)$, where $e$ is the element whose frequency we want to track.

The pseudocode of \CEN ~is presented in Appendix \ref{app:CEN}.

\vspace{3mm}

\noindent{\bf \CE.}
Huang et al.\ \cite{HYZ12} proposed a randomized algorithm \CE\ with a better performance. We summarize their main result in the following lemma. 

\begin{lemma}[\cite{HYZ12}]
  \label{lem:CE}
Given an element $e$, $\text{\CE}(e, \eps,\delta)$ maintains a $(1+\epsilon, \delta)$-approximation to $e$'s frequency at the coordinator, using $O\left( (k + \frac{\sqrt{k}}{\epsilon})\log \frac{1}{\delta} \log^2 m \right)$ bits of communication, $O(\log m \log \frac{1}{\delta})$ bits space per site and amortized $O(\log \frac{1}{\delta})$ processing time per item.
\end{lemma}
For $V \subseteq [n]$, \CE$(V, \eps, \delta)$ maintains a $(1 + \eps, \delta)$-approximation of $m_V = \sum_{i\in V}m_i$, the total frequencies of elements in $V$. Similarly,  \CEN$(V, \eps)$ maintains a $(1 + \eps)$-approximation of $m_V$.
\vspace{3mm}

\subsection{The Algorithms}
To describe the algorithms, we need to introduce a positive ``constant'' $\lambda$ which depends on the property of the function to be tracked. As mentioned that different from the streaming model where we can maintain $R$ exactly, in distributed monitoring we can only maintain an approximation of $S$'s tail frequency $R$. For a function $f$, recall that $\bar{f}(A) = \E[X] = \E[f(R) - f(R-1)]$.  The observation is that,  if $\hat{X} = f(\hat{R})-f(\hat{R}-1)$ is very close to $X = f(R) - f(R-1)$ when $\hat{R}$ is close to $R$, then $\Est(f, \hat{R}, \kappa)$ will be a relatively accurate estimation of $\Est(f, R, \kappa)$ (hence $\E[X]$).  To be more precise, if $\hat{R}\in \mathbb{Z}^+$ is a $(1 + \eps)$-approximation to $R$, we hope $|X - \hat{X}|$ can be bounded by
\begin{equation}
  \label{eq:trackable-ineq}
  |X - \hat{X}| \leq \lambda \cdot \eps \cdot X.  
\end{equation}

Unfortunately, for some functions there is no such $\lambda$. For example, let us consider $f(x)= x \log \frac{m}{x}$ (the function for the Shannon entropy)  and  $A = \{1,1, \ldots, 1\}$. If (\ref{eq:trackable-ineq}) holds for some positive $\lambda$, we have $\E[|X - \hat{X}|] \leq \lambda \cdot \eps \cdot\E[X] = \lambda \cdot \eps \cdot\bar{f}(A) = 0$, which is clearly not possible.  

To fix above issue, we can get rid of ``bad inputs'' (we can handle them using other techniques) by putting our discussion of $\lambda$ under a restricted input class. That is, the constant $\lambda$ depends on both the function $f$ and the set of possible inputs $\mathcal{A}$. Formally, we introduce $\lambda_{f, \mathcal{A}}$ (the subscript emphasizes that $\lambda$ depends on both $f$ and $\mathcal{A}$) as following,  

\begin{definition}[$\lambda_{f, \mathcal{A}}$]
\label{def:trackable}
Given a function $f:\mathbb{N} \to \mathbb{R}^+ \cup \{0\}$ with $f(0) = 0$ and a class of inputs $\mathcal{A}$, we define $\lambda_{f, \mathcal{A}}$ be the smallest $\lambda$ that satisfies the following: 

\begin{itemize}
\item $\lambda \geq 1$,
\item for any $A \in \mathcal{A}$, let $(S, R) \sAMS A$, for any positive number $\eps \le 1/4$ and any $\hat{R}$ that is a $(1+\eps)$-approximation of $R$, we have 
\begin{equation}
  \label{eq:def:trackable}
 |X - \hat{X}|  \le \lambda \cdot \eps \cdot X,
\end{equation}
where $X$ and $\hat{X}$ equal $f(R) - f(R-1)$ and $f(\hat{R}) - f(\hat{R} - 1)$ respectively.
\end{itemize}
\end{definition}

When $f$ and $\mathcal{A}$ are clear from the context, we often write $\lambda_{f, \mathcal{A}}$ as $\lambda$. $\lambda$  measures the approximation error introduced by using the approximation $\hat{R}$ when  estimating $\E[X] = \bar{f}(A)$ under the \emph{worst-case} input $A \in \mathcal{A}$. We will see soon that the efficiency of our AMS sampling algorithm is directly related to the value of $\lambda$. 

\begin{remark}
Directly calculating $\lambda$ based on $f$ and $\mathcal{A}$ may not be easy, but for the purpose of bounding the complexity of the AMS sampling, it suffices to calculate a relatively tight upper bound of $\lambda_{f, \mathcal{A}}$; examples can be found in Section~\ref{sec:analysis} and Section~\ref{sec:Tsallis} when we apply this algorithm framework to entropy functions.
\end{remark}

We now show how to maintain a single pair $(S, \hat{R})\sAMS A$ (we use $\hat{R}$ because we can only track $R$ approximately). The algorithms are presented in Algorithm~\ref{algo:ReceiveItemAMS} and \ref{algo:SampleUpdateAMS}.
\medskip

\begin{algorithm}[H]
intialize $S = \perp, r(S) = +\infty$\;
\ForEach{$e$ received}{
  sample $r(e) \in_R (0,1)$\;
  \lIf {$r(e) < r(S)$}{
      send $(e, r(e))$ to the coordinator
    }
}
\caption{ Receive an item at a site}
\label{algo:ReceiveItemAMS}
\end{algorithm}

\begin{algorithm}[H]
\ForEach {$(e, r(e))$ received}{
  update $S \gets e$, $r(S) \gets r(e)$\;
  restart $\hat{R} \leftarrow \text{\CEN}(S, \frac{\eps}{3\lambda})$\;
  broadcast new $(S, r(S))$ to all sites and each site updates their local copy.
}
\caption{ Update a sample at the coordinator}
\label{algo:SampleUpdateAMS}
\end{algorithm}

\begin{itemize}
\item {\em Maintain $S$}: Similar to that in \cite{CMYZ12}, we randomly associate each incoming item  $a$ with a real number~\footnote{In practice, one can generate a random binary string of, say, $10\log m$ bits as its rank, and w.h.p. all ranks will be different.} $r(a) \in (0,1)$ as its rank. We maintain $S$ to be the item with the smallest rank in $A(t)$ at any time step $t$. Each site also keeps a record of $r(S)$, and only sends items with ranks smaller than $r(S)$ to the coordinator. Each time $S$ getting updated, the coordinator broadcasts the new $S$ with its rank $r(S)$ to all the $k$ sites.

\item {\em Maintain $R$}: Once $S$ is updated, we use \CEN$(S, \frac{\eps}{3 \lambda})$ to keep track of its tail frequency $R$ up to a factor of $(1+\frac{\eps}{3\lambda})$.

\end{itemize}

To present the final algorithm, we need to calculate $\kappa$, the number of copies of $(S, \hat{R}) \sAMS A$ we should maintain at the coordinator. Consider a fixed function $f$ and an input class $\mathcal{A}$. Recall that in Lemma \ref{lem:c}, $a, b$ and $\E[X]$ all depend on the input stream $A \in \mathcal{A}$ because the distribution of $X$ is determined by the input stream. To minimize the communication cost, we want to keep $\kappa$ as small as possible while Inequality (\ref{eq:inq:c}) holds for all input streams in $\mathcal{A}$. Formally, given an input stream $A\in\mathcal{A}$, we define $\pi(A)$ as
\begin{equation}
\label{eq:pi}
\begin{aligned}
& \underset{a, b}{\text{minimize}}
& & \frac{3(1 + a/\E[X])^2(a + b)}{(a + \E[X])}\\
& \text{subject to}
& & a \geq 0, \\
&&&  b > 0,\\
&&& -a \leq X \leq b ~(\forall X).
\end{aligned}
\end{equation}
Then $\kappa$ takes the upper bound of $\eps^{-2}\ln(2\delta^{-1})\pi(A)$ over $A\in\mathcal{A}$, that is, 
\begin{equation}
\label{eq:kappa}
    \kappa(\eps, \delta, \mathcal{A}) = \eps^{-2}\ln(2\delta^{-1}) \cdot \sup_{A\in\mathcal{A}} \pi(A).
\end{equation}

One way to compute $\sup_{A\in\mathcal{A}}\pi(A)$, as we will do in the proof of Lemma \ref{lem:entropy-1}, is to find specific values for $a$ and $b$ such that under arbitrary stream $A\in\mathcal{A}$, $ -a \leq X \leq b$ holds for all $X$. We further set $E = \inf_{A\in\mathcal{A}}\E[X]$, then an upper bound of $\sup_{A\in\mathcal{A}}\pi(A)$ is given by $O(\frac{(1 + a/E)^2(a + b)}{(a + E)})$. 


Our algorithm then maintains $\kappa = \kappa(\frac{\eps}{2}, \delta, \mathcal{A})$ copies of $(S, \hat{R}) \sAMS A$ at the coordinator. At each time step, the coordinator computes \Est$(f, \hat{R}, \kappa)$. We present the main procedure  in Algorithm \ref{algo:TrackingF}.
\medskip

\begin{algorithm}[H]
\tcc*[h]{$(S, \hat{R})\sAMS A$ are maintained via Algorithm \ref{algo:ReceiveItemAMS}, \ref{algo:SampleUpdateAMS}}

track $\kappa(\frac{\eps}{2}, \delta, \mathcal{A})$ (defined in Equation~(\ref{eq:kappa})) copies of $(S,\hat{R})\sAMS A$ in parallel\;
\Return the average of all $\left( f(\hat{R}) - f(\hat{R} - 1) \right)$\; 
\caption{{\bf TrackF$(\eps, \delta)$}: Track $\bar{f}(A)$ at the coordinator}
\label{algo:TrackingF}
\end{algorithm}

\subsection{The Analysis}
We prove the following result in this section.
\begin{theorem}
\label{thm:ams-inf}
For any function $f:\mathbb{N} \to \mathbb{R}^+ \cup \{0\}$ with $f(0) = 0$ and input class $\mathcal{A}$, Algorithm \ref{algo:TrackingF} maintains at the coordinator a $(1 + \eps, \delta)$-approximation to $\bar{f}(A)$ ~for any $A \in \mathcal{A}$, using 
\begin{equation}
  \label{eq:inf-cost}
   O \left (k/\eps^3 \cdot \lambda\cdot \sup_{A\in\mathcal{A}}\pi(A)\cdot \log \frac{1}{\delta}\log^2 m\right )
\end{equation}
 bits of communication, $O\left(\kappa(\frac{\eps}{2}, \delta, \mathcal{A})\cdot\log m \right)$ bits space per site, and amortized $O\left(\kappa(\frac{\eps}{2}, \delta, \mathcal{A})\right)$ time per item, where $\pi(A), \kappa(\frac{\eps}{2}, \delta, \mathcal{A})$ are defined in (\ref{eq:pi}) and (\ref{eq:kappa}) respectively.
\end{theorem}

We first show the correctness of Algorithm \ref{algo:TrackingF}, and then analyze the costs.

\paragraph{Correctness.}
The following lemma together with the property of $\hat{R}$ gives the correctness of Algorithm~\ref{algo:TrackingF}.
\begin{lemma}
\label{lem:inf-correctness} For any $f : \mathbb{N} \to \mathbb{R}^+ \cup \{0\}$ with $f(0) = 0$ and input class $\mathcal{A}$, set $\kappa = \kappa(\frac{\eps}{2}, \delta, \mathcal{A})$. If $\hat{R}$ is a $(1 + \frac{\eps}{3\lambda})$-approximation to $R$, then \Est$(f, \hat{R}, \kappa)$ is a $(1 + \eps, \delta)$-approximation to $\bar{f}(A), \forall A \in \mathcal{A}$.
\end{lemma}

\begin{proof}
By Definition~\ref{def:trackable}, the fact ``$\hat{R}$ is a $(1 + \frac{\eps}{3\lambda})$-approximation to $R$'' implies   $|X - \hat{X}| \leq \frac{\eps}{3} X$, hence
\begin{equation}
  \label{eq:delta1}
  |\Est(f, R, \kappa) - \Est(f, \hat{R}, \kappa)| \leq \frac{\eps}{3}\Est(f, R, \kappa)
\end{equation}
By Lemma \ref{lem:c}, our choice for $\kappa$ ensures that $\Est(f, R, \kappa)$ is a $(1+\eps/2, \delta)$-approximation to $\E[X]$, that is, with probability at least $1-\delta$, we have 
\begin{equation}
  \label{eq:delta2}
  |\Est(f, R, \kappa) - \E[X]| \leq \frac{\eps}{2} \E[X].
\end{equation}
Combining (\ref{eq:delta1}) and (\ref{eq:delta2}), we obtain 
$$|\Est(f, \hat{R}, \kappa) - \E[X]| \leq \left(\left(1+\frac{\eps}{2}\right) \left(1 + \frac{\eps}{3}\right) - 1\right) \E[X] \leq \eps \E[X].$$ We thus conclude that $\Est(f, \hat{R}, \kappa)$ is a $(1 + \eps, \delta)$-approximation to $\E[X]$ for any input stream $A \in \mathcal{A}$.
\end{proof}

\paragraph{Costs.} By \CEN, tracking $\hat{R}$  as a $(1 + \eps)$-approximation to $R$ for each sample $S$ costs $O(\frac{k}{\eps}\log m)$ bits. We show in the following technical lemma (whose proof we deferred to Appendix \ref{app:lem:bd}) that the total number of updates of $S$ is bounded by $O(\log m)$ with high probability. Thus the total bits of communication to maintain one copy of $\hat{R}$ can be bounded by $O(\frac{k}{\eps}\log^2 m)$. 

\begin{lemma}
\label{lem:bd}
Let $U_1, \ldots, U_m$ be random i.i.d samples from $(0,1)$. 
Let $J_1 = 1$; for $i \ge 2$, let $J_i = 1$ if $U_i < \min\{U_1, \ldots, U_{i-1}\}$ and $J_i = 0$ otherwise. Let $J = \sum_{i \in [m]} J_i$. Then  $J_1,J_2,\ldots,J_m$ are independent,
and $\Pr[J > 2 \log m] < m^{-1/3}$.
\end{lemma}

We will ignore the failure probability $m^{-1/3}$ in the rest of the analysis since it is negligible in all cases we consider. 







\smallskip

We now bound the total communication cost: we track $\kappa(\frac{\eps}{2}, \delta, \mathcal{A})$ (defined in Equation~(\ref{eq:kappa})) copies of $(S,\hat{R})\sAMS A$ in parallel; and to maintain each such pair, we may restart \CEN\ for $O(\log m)$ times. Recall that the communication cost of each run of \CEN\ is $O(\frac{k \cdot\lambda}{\eps}\log m)$. The total communication cost (\ref{eq:inf-cost}) follows immediately. The space and processing time per item follows by noting the fact that maintaining each copy of $(S, \hat{R})\sAMS A$ needs $O(\log m)$ bits, and each item requires $O(1)$ processing time.   We are done with the proof of Theorem~\ref{thm:ams-inf}.




\bigskip

We can in fact use \CE\ instead of \CEN\ in Algorithm \ref{algo:SampleUpdateAMS} to further reduce the communication cost. The idea is straightforward: 
we simply replace \CEN$(S, \frac{\eps}{3\lambda})$ in Algorithm \ref{algo:SampleUpdateAMS} with \CE$(S, \frac{\eps}{3\lambda}, \frac{\delta}{2\kappa})$. 

\begin{corollary}
  \label{cor:ams-inf-2}
For any function $f:\mathbb{N} \to \mathbb{R}^+ \cup \{0\}$ with $f(0) = 0$ and input class $\mathcal{A}$, there is an algorithm that maintains at the coordinator a $(1+\eps, \delta)$-approximation to $\bar{f}(A)$, and it uses
\begin{equation}
  \label{eq:inf-cc}
   O \left ( \left(k/\eps^2 + \sqrt{k}/\eps^3\right)  \cdot \lambda\cdot \sup_{A\in\mathcal{A}}\pi(A)\cdot\log \frac{1}{\delta} \log \frac{\kappa}{\delta} \log^3m \right )
\end{equation}
 bits of communication, $O(\kappa \log m \log \frac{\kappa}{\delta})$  bits space per site, and amortized $O(\kappa\log \frac{\kappa}{\delta})$ time per item, where $\pi(A)$ and $\kappa = \kappa(\frac{\eps}{2}, \frac{\delta}{2}, \mathcal{A})$ are defined in (\ref{eq:pi}) and (\ref{eq:kappa}) respectively.
\end{corollary}

\begin{proof}
We track $\kappa(\frac{\eps}{2}, \frac{\delta}{2}, \mathcal{A})$ copies of $\hat{R}$ at the coordinator. Each $\hat{R}$ is tracked by  \CE$(S, \frac{\eps}{3\lambda}, \frac{\delta}{2\kappa})$  so that all $\kappa$ copies of $\hat{R}$ are still $(1 + \frac{\eps}{3\lambda})$-approximation to $R$ with probability at least $(1 - \frac{\delta}{2})$. Following the same arguments as that in Lemma \ref{lem:inf-correctness}, $\Est(f, \hat{R}, \kappa)$ will be a $(1 + \eps, \delta)$-approximation to $\E[X]$.

For the communication cost, recall that the communication cost of each \CE$(S, \frac{\eps}{3\lambda}, \frac{\delta}{2\kappa})$ is $O( (k + \frac{\sqrt{k} \lambda}{\epsilon})\cdot \log \frac{\kappa}{\delta} \log^2 m) = O\left( (k + \frac{\sqrt{k}}{\eps})\cdot \lambda\log \frac{\kappa}{\delta} \log^2 m\right)$ bits (Lemma~\ref{lem:CE}). Since we run $\kappa$ (defined in (\ref{eq:kappa})) copies of \CE\ and each may be restarted for $O(\log m)$ times, the total communication cost is bounded by (\ref{eq:inf-cc}). Each \CE$(S, \frac{\eps}{3\lambda})$ uses $O(\log m\log \frac{\kappa}{\delta})$ space per site and $O(\log \frac{\kappa}{\delta})$ processing time per item. We get the space and time costs immediately.
\end{proof}

\subsection{Sequence-Based Sliding Window}
\label{sec:ams-sliding}
In the sequence-based sliding window case, we are only interested in the last $w$ items received by the system, denoted by $A^w(t) = \{a_j\ |\ j > t - w\}$.

It is easy to extend the AMS sampling step to the sliding window case. Cormode et al.~\cite{CMYZ12} gave an algorithm that maintains $s$ random samples at the coordinator in the sequence-based sliding window setting. This algorithm can be directly used in our case by setting $s = 1$. Similar as before, when the sample $S$ is updated, we start to track its tailing frequency $R$ using \CE. The algorithm is depicted in Algorithm \ref{algo:TrackF-SW}.
\medskip

\begin{algorithm}[H]
$\kappa \gets \kappa(\frac{\eps}{2}, \frac{\delta}{2}, \mathcal{A})$\;
use the sequence-based sliding window sampling algorithm from \cite{CMYZ12} to maintain $\kappa$ independent samples\;\label{line:sampling}
each sample $S$ initiates a \CE$(S, \frac{\eps}{3\lambda}, \frac{\delta}{2\kappa})$ to track a $\hat{R}$. Whenever $S$ is updated, restart \CE \;
\Return the average of all $\left( f(\hat{R}) - f(\hat{R} - 1) \right)$ \;
\caption{{\bf TrackF-SW}$(\eps, \delta)$: Track $\bar{f}(A^w)$ in sequence-based sliding window setting}
\label{algo:TrackF-SW}
\end{algorithm}

\begin{theorem}
\label{thm:sliding}
For any function $f:\mathbb{N} \to \mathbb{R}^+ \cup \{0\}$ with $f(0) = 0$ and input class $\mathcal{A}$, let $\pi$ be defined as in (\ref{eq:pi}) but with $A$ being replaced with $A^w$. Let $\kappa = \kappa(\frac{\eps}{2}, \frac{\delta}{2}, \mathcal{A})$. There is an algorithm for the sequence-based sliding window (with window size $w$) that maintains at the coordinator a $(1 + \eps, \delta)$-approximation to $\overline{f}(A^w)$, using 
\begin{equation*}
  \label{eq:sliding-cost}
   O \left (\left(k/\eps^2 + \sqrt{k}/\eps^3\right) \cdot \lambda\cdot \sup_{A\in\mathcal{A}}\pi(A^w)\cdot \log \frac{1}{\delta}\log\frac{\kappa}{\delta}\log^3 m \right ) 
\end{equation*}
bits of communication, $O(\kappa \log m \log \frac{\kappa}{\delta})$  bits space per site, and amortized $O(\kappa \log \frac{\kappa}{\delta})$ time per item. 
\end{theorem}

\begin{proof}
In \cite{CMYZ12} it is shown that $O(k \log w \log m) = O(k \log^2 m)$ bits of communication is sufficient to maintain a random sample in $A^w$, and each site uses $O(\log m)$ bits space and $O(1)$ processing time per item. The rest of the proof is exactly the same as Corollary~\ref{cor:ams-inf-2}.
\end{proof}

\subsection{Time-Based Sliding Window}
In the time-based sliding window case, we are only interested in the items received in the last $t$ time steps, denoted by $A^t$.  

The algorithm of tracking $\bar{f}(A^t)$ is essentially the same as that in the sequence-based sliding window case (Algorithm \ref{algo:TrackF-SW}), except that in Line \ref{line:sampling} of Algorithm \ref{algo:TrackF-SW}, we use the time-based sampling algorithm from \cite{CMYZ12} instead of the sequence-based sampling algorithm. We summarize the result in the following theorem. Note that compared with Theorem~\ref{thm:sliding}, the only difference  is the extra $\log m$ in the space per site, which is due to the extra $\log m$ factor in the sampling algorithm for the time-based sliding window in \cite{CMYZ12}.

\begin{theorem}
\label{thm:sliding-time}
For any function $f:\mathbb{N} \to \mathbb{R}^+ \cup \{0\}$ with $f(0) = 0$ and input class $\mathcal{A}$, let $\pi$ be defined as in (\ref{eq:pi}) but with $A$ being replaced with $A^t$. Let $\kappa = \kappa(\frac{\eps}{2}, \frac{\delta}{2}, \mathcal{A})$. There is an algorithm for the time-based sliding window (with window size $t$) that maintains at the coordinator a $(1 + \eps, \delta)$-approximation to $\bar{f}(A^t)$, using 
\begin{equation*}
  \label{eq:sliding-cost-time}
   O \left (\left(k/\eps^2 + \sqrt{k}/\eps^3\right) \cdot \lambda\cdot \sup_{A\in\mathcal{A}}\pi(A^t)\cdot \log \frac{1}{\delta}\log\frac{\kappa}{\delta}\log^3 m \right ) 
\end{equation*}
bits of communication, $O(\kappa \log^2 m \log \frac{\kappa}{\delta})$  bits space per site, and amortized $O(\kappa \log \frac{\kappa}{\delta})$ time per item. 
\end{theorem}



%% file: entropy.tex

\section{Shannon Entropy}
\label{sec:Shannon}
In the  Shannon entropy function we have $f(x) = x \log \frac{m}{x}\ (x > 0)$ and $f(0) = 0$, where $m = \abs{A(t)}$. Let $f_m$ denote this function. In this section, we will show that for arbitrary input $A \in [n]^m$, we can track $\bar{f_m}(A) = \sum_{i\in[n]}\frac{m_i}{m}\log\frac{m}{m_i}$ efficiently, using only $\tilde{O}\left( k/\eps^2 + \sqrt{k}/\eps^3 \right)$ bits of communication. We also obtain similar results for sliding window cases.

To do this, we first show that when only considering a restricted input class $\mathcal{A}'$, $\bar{f_m}(A)\ (A \in \mathcal{A}')$ can be tracked efficiently by directly applying the AMS framework presented in previous section. We then discuss how to track $\bar{f_m}(A)$ under arbitrary input $A \in [n]^m$. 

For technical reasons, we assume $1/m \leq \delta, \eps \leq 1/20$ throughout this section. As mentioned, in distributed monitoring we can only maintain a $(1+\eps)$-approximation of $m$ at the coordinator using $o(m)$ bits of communication, but for the sake of simplifying the presentation, we assume that $m$ can be maintained at the coordinator {\em exactly without any cost}. Appendix \ref{app:why:m} explains why we can make such an assumption. The same assumption is also applied to the analysis of the Tsallis Entropy in Section~\ref{sec:Tsallis}.
\vspace{3mm}

\noindent{\bf Roadmap.} We will again start by introducing some tools from previous work.  We then define the restricted input class $\mathcal{A}'$, and give some intuition on how to track general inputs.  We next give the algorithm for the infinite window case, followed by the analysis.  Finally we discuss the sliding window cases.

\subsection{Preliminaries}
To present our algorithm for the Shannon entropy we need a few more tools from previous work.
\vspace{2mm}

\noindent{\bf\CA.}
Yi and Zhang \cite{YZ09} gave a deterministic algorithm, denoted by \CA$(\eps)$, that can be used to track the empirical probabilities of all universe elements up to an additive approximation error $\eps$ in distributed monitoring. 
We summarize their main result below.

\begin{lemma}[\cite{YZ09}]
  \label{lem:CA}
For any $0 < \eps \leq 1$, \CA$(\eps)$ uses  $O(\frac{k}{\eps}\log^2 m)$ bits of communication,
such that for any element $i \in [n]$, it maintains at the coordinator an estimation $\hat{p_i}$ such that
$\abs{\hat{p_i} - p_i} < \eps$. Each site uses $O(\frac{1}{\eps}\log m)$ bits space and amortized $O(1)$ time per item.
\end{lemma}

\noindent{\bf\CM.}
We will also need  the  CountMin sketch introduced by Cormode and Muthukrishnan \cite{cormode05}  in the streaming model.  We summarize its property below.
\vspace{-1mm}
\begin{lemma}[\cite{cormode05}]
  \label{lem:CM}
  The \CM$(\eps,\delta)$  sketch uses $O(\frac{1}{\eps}\log m \log \frac{1}{\delta})$ bits of space in the streaming model, such that for any given element $i \in [n]$, it gives an estimation $\hat{m}_i$ of $i$'s frequency $m_i$ such that
    $\Pr[m_i \leq \hat{m}_i \leq m_i + \eps m_{-i}] \ge 1 - \delta$.
The processing time for each item is $O(\log \frac{1}{\delta})$.
\end{lemma}

\subsection{ Tracking $f_m$ Under A Restricted Class $\mathcal{A}'$}
We have briefly mentioned (before Definition \ref{def:trackable}) that if we consider all possible inputs, $f_m = x\log \frac{m}{x}$ cannot be tracked efficiently by directly using our AMS sampling framework because the corresponding $\lambda$ does not exist.
However, if we consider another input class $$\mathcal{A}'=\{ A \in [n]^{m'}: 0 < m' \leq m, \forall i\in[n], m_i \leq 0.7 m\},$$ 
(in other words, we consider streams with length no more than $m$ and the frequency of each element is bounded by $0.7m$),  then we can upper bound $\lambda_{f_m, \mathcal{A}'}$ by a constant.

The following two lemmas show that under input class $\mathcal{A}'$, $f_m$ can be tracked efficiently using the AMS framework in Section~\ref{sec:AMS}. 

\begin{lemma}
  \label{lem:trackable}
Let $f_m$ and the input class $\mathcal{A}'$ be defined above. We have $\lambda_{f_m, \mathcal{A}'} \le 10$ and $\inf_{A' \in \mathcal{A}'} \bar{f_m}(A') \ge 0.5$.
\end{lemma}
\begin{proof}
Let $r,\hat{r} \in \mathbb{Z}^+$, where $\hat{r}$ is a $(1+\eps)$-approximation to $r$. Let $X = X(r) = f_m(r) - f_m(r-1)$ and $\hat{X} = X(\hat{r})$.  Taking the derivative, $X'(r) = f_m'(r) - f_m'(r-1) = -\log \left(1+\frac{1}{r-1}\right) < 0$, and thus $\inf X = f_m(0.7m) - f_m(0.7m-1) \overset{m\gg 1}{\approx} \log 0.7^{-1} > 0.5$. 
  
When $r\ge 2$, we have
$$\abs{X(r) - X(\hat{r})} \le \eps r \log \left(1 + \frac{1}{(1-\eps)r - 1}\right)
                          \le 5 \eps;
$$
and when $r=1$, we have $\hat{r} = r$ hence $X = \hat{X}$. 
Therefore $\abs{X - \hat{X}} \le 5 \eps \le 10\cdot \eps \cdot X$ (as $\inf X > 0.5$). Consequently we can set $\lambda = 10$, and thus $\lambda_{f_m, \mathcal{A}'} = \inf \{ \lambda\}\leq 10$.

Next, given any $A'\in \mathcal{A}'$, we have $m_i < 0.7 m$ for all $i\in[n]$, 
 and thus $\bar{f_m}(A')=\frac{1}{|A'|}\sum_{i\in[n]}f_m(m_i) > \log 0.7^{-1}  > 0.5.$
\end{proof}

\begin{lemma}
\label{lem:entropy-1}
Let $f_m$ and $\mathcal{A}'$ be defined above. Algorithm \ref{algo:TrackingF} (with \CEN\ in Algorithm~\ref{algo:ReceiveItemAMS} and Algorithm \ref{algo:SampleUpdateAMS} replaced by \CE)  maintains a $(1+\eps, \delta)$-approximation to $\bar{f_m}(A)$ for any $A \in \mathcal{A}'$ at the coordinator, using 
$O\left( (k/\eps^2 + \sqrt{k}/\eps^3)\cdot \log\frac{1}{\delta}\log^4 m \right)$ bits of communication, $O(\eps^{-2}\cdot \log^3 m \log\frac{1}{\delta})$  bits space per site, and amortized $O(\eps^{-2}\cdot\log\frac{1}{\delta}\log^2 m)$ time per item.

\end{lemma}

\begin{proof}
This lemma is a direct result of Corollary \ref{cor:ams-inf-2}.
The main task to derive the stated is to bound $\sup_{A\in \mathcal{A}'}\pi(A)$. Recall that $X(R) = f_m(R) - f_m(R-1)$, as $X'(R) < 0$ we have $0.5 < X(0.7m) \leq X \leq X(1) = \log m$. To give an upper bound of $\kappa$, it suffices to use $a = 0, b = \log m$, and set $E = \inf_{A\in\mathcal{A'}}\bar{f_m}(A) \ge 0.5$, which gives an upper bound $$\sup_{A\in\mathcal{A'}}\pi(A)\leq \frac{3(1 + a/E)^2(a + b)}{(a + E)} = O(\log m),$$ or $\kappa(\frac{\eps}{2}, \frac{\delta}{2}, \mathcal{A}') = O(\eps^{-2} \log m \log \delta^{-1})$.

Thus we maintain $\Theta(\eps^{-2} \log m \log \delta^{-1})$ copies of estimators at the coordinator. The lemma follows by applying Corollary \ref{cor:ams-inf-2} with the values $\lambda_{f_m, \mathcal{A}'}$ and $\sup_{A\in\mathcal{A}'}\pi(A)$ chosen above (and note $O(\log \frac{\kappa}{\delta}) = O(\log m)$). 
\end{proof}

\subsection{Intuition on Tracking $f_m$ under $\mathcal{A} = [n]^m$}

To track $f_m$ under input class $\mathcal{A} = [n]^m$, a key observation made by Chakrabarti et al.\ \cite{CCM10} is that we can use the following expression to compute the entropy of $A$ when the stream $A\in \mathcal{A}$ has a frequent element $z$ (say, $p_z \ge 0.6$):
 \begin{eqnarray}
    H(A) &=& \frac{1}{m}\sum_{i=1}^n f_m(m_i) \nonumber \\
    &=& (1-p_z) \E[X'] + p_z\log ({1}/{p_z}) \label{eq:HA},
 \end{eqnarray}
where $X' = f_m(R') - f_m(R' - 1)$, and $(S', R') \sAMS A\backslash z$.  Note that $A\backslash z \in \mathcal{A}'$. Thus by Lemma \ref{lem:entropy-1} we can track $\E[X'] = \bar{f_m}(A\backslash z)$ efficiently.

We try to implement this idea in distributed monitoring. The remaining tasks are: (1) keep track of the pair $(S', R')$ (thus $X'$); and (2) keep track of $(1 - p_z)$ and $p_z$. Compared with the streaming model \cite{CCM10}, both tasks in distributed monitoring require some new ingredients in algorithms and analysis, which we present in Section~\ref{sec:algo} and Section~\ref{sec:analysis}, respectively.

\subsection{The Algorithms}
\label{sec:algo}

We first show how to maintain the pair $(S', R')$, $p_z$ and $1 - p_z$ approximately.

\paragraph{Maintain $(S', R')$.} As observed in \cite{CCM10}, directly sampling $(S', R')$ is not easy. The idea in \cite{CCM10} is to maintain $(S_0, R_0) \sAMS A$ and $(S_1, R_1) \sAMS A \backslash S_0$. We also keep track of the item $z$ with $p_z \ge 0.6$, if exists. Now we can construct $(S', R')$ as follows: if $z \neq S_0$, then $(S',R') \gets (S_0, R_0)$; otherwise $(S',R') \gets (S_1, R_1)$. The proof of the fact that $S'$ is a random sample from $A \backslash z$ can be found in \cite{CCM10}, Lemma 2.4. 
Algorithm~\ref{algo:ReceiveItem} and \ref{algo:HandleMessageC} show how to maintain $S_0, S_1$.  Algorithm~\ref{algo:countR} (that is, \EstR$(\eps, \delta_1)$, where $\delta_1$ will be set to $\delta/4\kappa$ in Algorithm~\ref{algo:entropyEst}) shows how to maintain $\left(1+\frac{\eps}{\lambda_{f_m, \mathcal{A}'}}, \delta_1\right)$-approximations to $R_0$ and $R_1$,
 and consequently a $\left(1+\frac{\eps}{\lambda_{f_m, \mathcal{A}'}}, \delta_1\right)$-approximation to $R'$, which guarantees that $|X' - \hat{X}'| \leq \eps \cdot X'$ holds with probability at least $(1 - \delta_1)$. 
\medskip

\begin{algorithm}[H]
initialize $S_0 = S_1 = \perp, r(S_0) = r(S_1) = +\infty$\;
\ForEach{$e$ received}{
sample $r(e) \in_R (0,1)$\;
  \If {$e = S_0$}{
    \lIf {$r(e) < r(S_0)$}{
      send the coordinator ``update $(S_0, r(S_0))$ with $(e, r(e))$''
    }
  }\ElseIf {$e \neq S_0$}{
    \If {$r(e) < r(S_0)$} {
      send the coordinator ``update $(S_1, r(S_1))$ with $(S_0, r(S_0))$''\;
      send the coordinator ``update $(S_0, r(S_0))$ with $(e, r(e))$''\;
    }
    \lElseIf {$r(e) < r(S_1)$}{
      send the coordinator ``update $(S_1, r(S_1))$ with $(e, r(e))$''
    }
  }
}
\caption{Receive an item at a site (for the Shannon entropy)}
\label{algo:ReceiveItem}
\end{algorithm}

\begin{algorithm}[H]
\ForEach {message \textbf{msg} received}{
 execute \textbf{msg}: update $(S_0, r(S_0))$ and/or $(S_1, r(S_1))$ based on \textbf{msg}\;
  broadcast \textbf{msg} to all sites and request each site to execute the \textbf{msg}\;
}
\caption{Update samples at the coordinator (for the Shannon entropy)}
\label{algo:HandleMessageC}
\end{algorithm}

\begin{algorithm}[H]
initialize $S_0 = S_1 = \perp$, and $r(S_0) = r(S_1) = +\infty$\; 
set $\eps_1 \leftarrow \frac{\eps}{10}$\;
  \lIf {$S_0$ is updated by the same element}{
    restart $\text{\CE}(S_0, \eps_1, \delta_1)$
  }
  \lElseIf {$S_0$ is updated by a different element}{
    restart $\text{\CE}(S_0, \eps_1, \delta_1)$
  }
  \ElseIf {$S_1$ is updated}{
    \If {$S_1$ is updated by $S_0$}{
      replace the whole data structure of $\text{\CE}(S_1, \eps_1, \delta_1)$ with $\text{\CE}(S_0, \eps_1, \delta_1)$\;
    }
    \lElse{
          restart $\text{\CE}(S_1, \eps_1, \delta_1)$
        }
  }

\caption{{\bf $\text{\EstR}(\eps, \delta_1)$}: Maintain $R_0$ and $R_1$ at the coordinator}
\label{algo:countR}
\end{algorithm}

\paragraph{Maintain $p_z$ and $1 - p_z$.} It is easy to use \CA\ to maintain $p_z$ up to an additive approximation error $\eps$, which is also a $(1+O(\eps))$-approximation of $p_z$ if $p_z \ge 0.6$. However, to maintain a $(1+\eps)$-relative approximation error of $(1 - p_z)$ is non-trivial when $(1 - p_z)$ is very close to $0$. We make use of \CA, \CEN\ and \CM\ to construct an algorithm \EstP$(\eps, \delta)$, which maintains a $(1+\eps, \delta)$-approximation of $(1 - p_z)$ at the coordinator when $p_z > 0.6$. We describe \EstP\ in Algorithm~\ref{algo:trackprob}.
\medskip

\begin{algorithm}[H]
initialize $z \leftarrow \perp$; ${ct} \leftarrow 0$; $\forall i \in [k], {ct}_i \leftarrow 0$\;
\tcc*[h]{run the following processes in parallel:}\\
run $\text{\CA}(0.01)$ \;
run $\gamma \gets \text{\CEN}(-z,\eps/4)$ \;
run $\hat{m} \gets \text{\CEN}([n], \eps/4)$\;

the coordinator maintains a counter ${ct}$ that counts the number of items received by all sites up to the last update of $z$\;

each site maintains a local \CM$(\eps/4, \delta)$ sketch\;

each site $S_i$ maintains a counter ${ct}_i$ that counts the number of items received at $S_i$\;

\tcc*[h]{monitored by $\text{\CA}(0.01)$}\\
\If{\CA\ identifies a new frequent element $e$ with $\hat{p}_e \ge 0.59$}{
$z \gets e$. Broadcast $z$ to all sites\;
 restart $\gamma \gets \text{\CEN}(-z, \eps/4)$\;
 each site $S_i$ sends its local \CM\ sketch and local counter $ct_i$ to the coordinator\;
the coordinator merges $k$ local \CM\ sketches to a global \CM, and sets ${ct} = \sum_{i \in [k]}{ct}_i$\;
}
\Return $z$, $\hat{p}_{-z} \gets \frac{ct - \text{\CM}[z] + \gamma}{\hat{m}}$; $ \hat{p}_z \gets 1 - \hat{p}_{-z}$.

\caption{{\bf \EstP$(\eps, \delta)$}: Approximate the empirical probability of a frequent element}
\label{algo:trackprob}
\end{algorithm}

\vspace{2mm}

\paragraph{Putting Things Together.} Let $(S_0, \hat{R}_0)$ and $(S_1, \hat{R}_1)$ be samples and their associated counts maintained by Algorithm~\ref{algo:ReceiveItem}, \ref{algo:HandleMessageC}, and \ref{algo:countR}.  
The final algorithm for tracking $H(A)$ is depicted in Algorithm~\ref{algo:entropyEst}.
\medskip

\begin{algorithm}[H]
$\kappa \gets 480 \eps^{-2}\ln (4\delta^{-1})(2+\log m)$\;
\BlankLine
\tcc*[h]{maintain $z$, $\hat{p}_z$ and $\hat{p}_{-z}$}\\
run $\text{\EstP}(\eps/4, \delta/2)$\;
\tcc*[h]{get $\kappa$ independent copies of $(S_0, \hat{R}_0, S_1, \hat{R}_1)$}\\
run $\kappa$ copies of Algorithm~\ref{algo:ReceiveItem}, \ref{algo:HandleMessageC}, and \EstR$(\eps/6, {\delta}/{(4\kappa)})$ in parallel\;

\If {$\hat{p}_z > 0.65$}{
  \tcc*[h]{For each copy of $(S_0, \hat{R}_0, S_1, \hat{R}_1)$, construct $\hat{R}'$}\\
  \lIf {$S_0 = z$}{
    $\hat{R}' \gets \hat{R}_1$
   }
   \lElse{
    $\hat{R}' \gets \hat{R}_0$
   }
   \tcc*[h]{$\text{\Est}(f_m, \hat{R}', \kappa)$ gives the average of $\kappa$ copies of $f_m(\hat{R}') - f_m(\hat{R}'-1)$}\\
   \Return $(1 - \hat{p}_z) \text{\Est}(f_m, \hat{R}', \kappa) + \hat{p}_z \log (1/\hat{p}_z)$\;
}
\Else {
  \Return ~\Est$(f_m, \hat{R}_0, \kappa)$
}

\caption{{\bf \EstH $(\eps, \delta)$}: Approximate the Shannon entropy}
\label{algo:entropyEst}
\end{algorithm}

\subsection{The Analysis}
\label{sec:analysis}
We prove the following result in this section.
\begin{theorem}
  \label{thm:ent-inf-cost}
  \EstH$(\eps,\delta)$ maintains at the coordinator a $(1+\eps, \delta)$-approximation to the Shannon entropy, using
$O\left( (k/\eps^2 + \sqrt{k}/\eps^3)\cdot \log\frac{1}{\delta}\log^5 m \right)$
bits of communication, $O(\eps^{-2} \cdot  \log \frac{1}{\delta}\log^3 m)$ bits space per site and amortized $O(\eps^{-2}\cdot \log \frac{1}{\delta} \log^2 m)$ time per item.
\end{theorem}

We first show the correctness of \EstH, and then analyze its costs. 
\vspace{3mm}

\noindent{\bf Correctness.}  We establish the correctness by the following two lemmas.  The first lemma shows that if there is a frequent element $z$ with empirical probability $p_z\geq 0.6$, then Algorithm~\ref{algo:trackprob} properly maintains $1 - p_z$.   The second lemma shows the correctness of Algorithm~\ref{algo:entropyEst} for any input $A \in [n]^m$.

\begin{lemma}
\label{lem:estP}
$\hat{p}_{-z}$ (see Line $13$ of Algorithm \ref{algo:trackprob}) is a $(1+\eps, \delta)$-approximation to $(1-p_z)$.
\end{lemma}

\begin{proof}
Let $z\ (p_z \ge 0.6)$ be the candidate frequent element if exists. Let $t(z)$ be the time step of the most recent update of $z$.
At any time step, let $A^0$ be the substream consisting of all items received on or before $t(z)$, $ct = \abs{A^0}$, and let $A^1$ be the rest of the joint stream $A$. Let $m_z^0$ and $m_{-z}^0$ be the frequency of element $z$ in $A^0$ and the sum of frequencies of elements other than $z$ in $A^0$, respectively. Similarly, let $m_z^1$ and $m_{-z}^1$ be defined for $A^1$. 

Algorithm \ref{algo:trackprob} computes $\hat{m}^1_{-z}$ as an approximation to $m^1_{-z}$  by $\text{\CEN}(-z, \eps/4)$, and $\hat{m}$ as an approximation to $m$ by $\text{\CEN}([n], \eps/4)$, both at the coordinator. The coordinator can also extract from the global \CM\ sketch an $\hat{m}_z^0$, which  approximates $m_z^0$ up to an additive approximation error $\frac{\eps}{4} m_{-z}^0$ with probability $(1 - \delta)$. At Line $13$ of Algorithm \ref{algo:trackprob}, at any time step, the coordinator can compute 
\begin{equation*}
1 - \hat{p}_z = \hat{p}_{-z} = \frac{ct - \hat{m}_z^0  + \hat{m}^1_{-z}}{\hat{m}},
\end{equation*}
where $\hat{m}_z^0$ is an $(\frac{\eps}{4} m_{-z}^0, \delta)$-approximation of $m_z^0$, and $\hat{m}^1_{-z}$ and $\hat{m}$ are $(1+\eps/4)$-approximation of $m^1_{-z}$ and $m$, respectively. 

The lemma follows by further noting that $ct = m^0_z + m^0_{-z}$. 
\end{proof}

\begin{lemma}
\label{lem:entropy-correctness}
  \EstH$(\eps, \delta)$ (Algorithm~\ref{algo:entropyEst}) correctly maintains at the coordinator a $(1+\eps, \delta)$-approximation to the empirical entropy $H(A)$.
\end{lemma}
\begin{proof}
Similar to \cite{CCM10} Theorem 2.5, we divide our proof into two cases.
\vspace{2mm}

\noindent{\bf Case 1:} there does not exist a $z$ with $\hat{p}_z \le 0.65$. We reach line 9 of Algorithm~\ref{algo:entropyEst}.
The assumption that $\eps < 1/20$ implies $p_z < 0.7$, and thus the input stream $A\in \mathcal{A}'$. It is easy to verify that our choices of parameters satisfy the premise of Lemma \ref{lem:inf-correctness}, thus the correctness.
\vspace{2mm}

\noindent{\bf Case 2:}  there is a $z$ with $\hat{p}_z > 0.65$. We reach line 7 of Algorithm~\ref{algo:entropyEst}. In this case we use Equation (\ref{eq:HA}).
The assumption that $\eps < 1/20$ implies $p_z > 0.6$, thus $A\backslash z \in \mathcal{A}'$. At line 3, we run \EstR$({\eps}/{6}, \frac{\delta}{4\kappa})$ so that with probability $1 -\frac{\delta}{4}$, the $\kappa$ copies $X'$s satisfy $|X' - \hat{X'}| \leq \frac{\eps}{6} X'$ simultaneously by a union bound.  One can verify based on (\ref{eq:kappa}) that our choice for $\kappa$ is large enough to ensure that \Est$(f_m, R', \kappa)$  is a $(1+\eps/4, \delta/4)$-approximation to $\E[X']$. Applying the same argument as in Lemma \ref{lem:inf-correctness}, we have $\Est(f_m, \hat{R}', \kappa)$ as a $(1 + \eps/2, \delta/2)$-approximation to  $\E[X']$.

At line 2, \EstP$(\eps/4, \delta/2)$ gives $(1-\hat{p}_z)$ as a $(1+\eps/4, \delta/2)$-approximation to $(1-p_z)$ (by Lemma \ref{lem:estP}). Further noting that when $(1 - \hat{p}_z)$ is a $(1+\eps/4)$-approximation to $(1 - p_z)$, we have
\begin{eqnarray*}
  &&\frac{|\hat{p}_z\log({1}/{\hat{p}_z}) - p_z \log({1}/{p_z})|}{p_z \log({1}/{p_z})} \\
   &\leq& \frac{|\hat{p}_z - p_z|}{p_z \log (1/p_z)} \max_{p\in [\frac{1}{2},1]}\abs{\frac{d(p\log{1}/{p})}{dp}} \\
   &\leq& \frac{\frac{\eps}{4} (1 - p_z)}{p_z\log (1/p_z)} \log e  \\
   & \leq& \eps.
\end{eqnarray*}
Thus $(1-\hat{p}_z)\text{\Est}(f_m, \hat{R}', \kappa) + \hat{p}_z\log \frac{1}{\hat{p}_z}$ is a $(1+\eps, \delta)$-approximation to $H(A) = (1-p_z)\E[X'] + p_z\log \frac{1}{p_z}$.
\end{proof}


\paragraph{Communication Cost.}  We shall consider the extra cost introduced by the following adaptations to the general AMS Sampling  framework described in Section \ref{sec:AMS}: (1) we need to maintain $S_0$ and $S_1$ rather than to simply maintain $S$; and (2) we have a new subroutine \EstP. It turns out that (1) will only introduce an additional multiplicative factor of $\log m$ to the communication, and (2) is not the dominating cost.


\smallskip

We first bound the total number of times $\text{\CE}(S_0)$ and $ \text{\CE}(S_1)$ being restarted.
\begin{lemma}
  \label{lem:sample2}
Let $C_0$ and  $C_1$ be the frequency of $\text{\CE}(S_0)$ and
$\text{\CE}(S_1)$ being restarted in Algorithm~\ref{algo:countR}, respectively. Then
 $(C_0+C_1)$ is bounded by $O(\log^2 m)$ with probability at least $1 - m^{-1/6}$. 
\end{lemma}

\begin{proof}
by Lemma \ref{lem:bd}, $C_0$ is bounded by $O(\log m)$ with probability at least $1 - m^{-1/3}$.
Now let us focus on $C_1$.
Suppose $n_1 < n_2 < \ldots < n_{C_0}$ are the global indices of items that update $S_0$. Let $b_i = r(a_{n_i})$, we have
$b_1 > b_2 > \ldots > b_{C_0}$. Let $A_i$ be the substream of $(a_{n_i}, a_{n_i+1}, \ldots a_{n_{i+1}-1})$ obtained by collecting all items that will be compared with $r(S_1)$; thus $\abs{A_i} \leq m$ and each item in $A_i$ is associated with a rank uniformly sampled from $(b_i, 1)$. For a fixed $C_0$, by  Lemma \ref{lem:bd} and a union bound we have that $C_1 < O(C_0\log m)$ with probability at least $1-\frac{C_0}{m^{1/3}}$. Also recall that $C_0 < 2\log m$ with probability $1 - m^{-1/3}$. Thus $C_1 = O(\log^2 m)$ with probability at least $1 -  \frac{2\log m}{m^{1/3}}$.
\end{proof}


We now bound the total communication cost.

\begin{lemma}
  \label{lem:proTrack} \EstP$(\eps, \delta)$ uses $O(\frac{k}{\eps}\log\frac{1}{\delta}\log^3 m)$  bits of communication.
\end{lemma}
\begin{proof}
We show that $z$ will be updated by at most $O(\log m)$ times. Suppose at some time step $m_0$ items have been processed, and $z = a$ is the frequent element. By definition, the frequency of $a$ must satisfy $m_a > 0.58m_0$. We continue to process the incoming items, and when $z$ is updated by another element at the moment the $m_1$-th item being processed, we must have $m_a < 0.42m_1$. We thus have $\frac{m_1}{m_0} \geq \frac{0.58}{0.42} = 1.38 > 1$, which means that every time $z$ gets updated, the total number of items has increased by a factor of at least $1.38$ since the last update of $z$. Thus the number of updates of $z$ is bounded by $O(\log m)$.
\smallskip



We list the communication costs of all subroutines in \EstP.
\begin{enumerate}
\item[(1)] $\text{\CA}(0.01)$ costs $O(k\log^2 m)$ bits;
\item[(2)] $\text{\CEN}([n], \eps/4)$ costs $O(\frac{k}{\eps}\log m)$ bits;
\item[(3)]  $\text{\CEN}(-z, \eps/4)$ costs $O(\frac{k}{\eps}\log m)$ bits; 
\item[(4)] sending $k$ sketches of \CM$(\eps/4, \delta)$  to the coordinator costs $O(\frac{k}{\eps}\log\frac{1}{\delta}\log m)$ bits;
\item[(5)] sending $k$ local counters to the coordinator costs $O(k\log m)$ bits. 
\end{enumerate}
Among them, (3), (4), (5) need to be counted by $O(\log m)$ times, and thus
the total communication cost is bounded by $O(\frac{k}{\eps}\log\frac{1}{\delta}\log^3 m)$.
\end{proof}

Combining Lemma \ref{lem:CE}, Lemma \ref{lem:entropy-correctness}, Lemma \ref{lem:sample2} and Lemma \ref{lem:proTrack}, we now prove Theorem \ref{thm:ent-inf-cost},

\begin{proof}
We have already showed the correctness and the communication cost.
  The only things left are the space and processing time per item. The processing time and space usage are dominated by those used to track $(S_0, \hat{R}_0, S_1, \hat{R}_1)$'s. So the bounds given in Lemma \ref{lem:entropy-1} also hold.

\end{proof}

\subsection{Sliding Windows}
\label{sec:entropy-sliding}

In Section \ref{sec:ams-sliding} we have extended our general AMS sampling algorithm to the sequence-based sliding window case. We can apply that scheme directly to the Shannon entropy. However, the communication cost is high when the Shannon entropy of the stream is small. On the other hand, it is unclear if we can extend the technique of removing the frequent element to the sliding window case: it seems hard to maintain $(S_0^w, R_0^w)\sAMS A^w$ and $(S_1^w, R_1^w)\sAMS A^w\backslash S_0^w$ simultaneously in the sliding window using $\text{poly}(k, 1/\eps, \log w)$ communication,  $\text{poly}(1/\eps, \log w)$ space per site and $\text{poly}(1/\eps, \log w)$ processing time per item.

By slightly adapting the idea in Section \ref{sec:ams-sliding}, we have the following result that may be good enough for most practical applications.
\begin{theorem}
\label{thm:entropy-sliding}
 There is an algorithm that maintains $\hat{H}(A^w)$ at the coordinator as an approximation to the Shannon entropy $H(A^w)$ in the sequence-based sliding window case such that $\hat{H}(A^w)$ is a $(1+\eps, \delta)$-approximation to $H(A^w)$  when $H(A^w) > 1$ and an $(\eps, \delta)$-approximation when $H(A^w) \leq 1$. The algorithm uses    $O \left ( (k/\eps^2 + \sqrt{k}/\eps^3) \cdot \frac{1}{\delta}\log^4 m \right )$
bits of communication, $O(\eps^{-2}\cdot\log \frac{1}{\delta}\log^3 m)$ bits space per site and amortized $O(\eps^{-2}\cdot\log\frac{1}{\delta}\log^2 m)$ time per item.
\end{theorem}

\begin{proof}
Instead of setting $\kappa$ (the number of sample copies we run in parallel) as Equation (\ref{eq:kappa}), we simply set $\kappa = \Theta(\eps^{-2}\log w\log \delta^{-1})$, thus the space and time usage for each site.
The correctness is easy to see: since we are allowed to have an additive approximation error $\eps$ (rather than $\eps \E[X]$) when $\E[X] \leq 1$, we can replace $\eps$ by $\frac{\eps}{\E[X]}$ in Inequality (\ref{eq:inq:c}) to cancel $\E[X]$. For the communication cost, we just replace the value of $\kappa$ in Section \ref{sec:ams-sliding} (defined by Equation (\ref{eq:kappa})) with $\Theta(\eps^{-2}\log w\log \delta^{-1})$.
\end{proof}

With the same argument we have a result for the time-based sliding window case where the window size is $t$. 

\begin{theorem}
\label{thm:entropy-sliding-time}
 There is an algorithm that maintains $\hat{H}(A^t)$ at the coordinator as an approximation to the Shannon entropy $H(A^t)$ in the time-based sliding window setting such that $\hat{H}(A^t)$ is a $(1+\eps, \delta)$-approximation to $H(A^t)$  when $H(A^t) > 1$ and an $(\eps, \delta)$-approximation when $H(A^t) \leq 1$. The algorithm uses \\ $O \left ( (k/\eps^2 + \sqrt{k}/\eps^3) \cdot \frac{1}{\delta}\log^4 m \right )$
bits of communication, $O(\eps^{-2}\cdot\log \frac{1}{\delta}\log^4 m)$ bits space per site and amortized $O(\eps^{-2}\cdot\log\frac{1}{\delta}\log^2 m)$ time per item.
\end{theorem}

\section{Tsallis Entropy}
\label{sec:Tsallis}
Recall that $\vec{p} = (p_1, p_2, \ldots, p_n) =  (\frac{m_1}{m}, \frac{m_2}{m}, \ldots, \frac{m_n}{m})$ is the vector of empirical probabilities. The $q$-th Tsallis entropy of a stream $A$ is defined as $$T_q(\vec{p}) = \frac{1-\sum_{i\in[n]}p_i^q}{q-1}.$$
It is well-known that when $q\rightarrow 1$, $T_q$ converges to  the Shannon entropy. In this section, we give an algorithm that continuously maintains a $(1+\eps, \delta)$-approximation to $T_q(\vec{p})$ for any constant $q > 1$.

Similar to the analysis for the Shannon entropy, we again assume that we can track the exact value of $m$ at the coordinator without counting its communication cost. To apply the general AMS sampling scheme, we use $g_m(x) = x - m (\frac{x}{m})^q$, hence 
\begin{eqnarray*}
\overline{g_m}(A) &=& \frac{1}{m}\sum_{i\in[n]}\left[m_i - m\left(\frac{m_i}{m}\right)^q\right] \\
& =& 1 - \sum_{i\in[n]}p_i^q = (q-1)T_q(\vec{p}).
\end{eqnarray*}
Let $Z$ consist of elements in the stream $A$ such that each $z \in Z$ has $m_z \ge 0.3m$. Thus $\abs{Z} \le 4$.
Consider the following two cases:
\begin{itemize}
\item $Z = \emptyset$. In this case  
$\overline{g_m}(A) \geq 1 - (\frac{1}{3})^{q-1}$.

\item $Z \neq \emptyset$. In this case 
$$\overline{g_m}(A) = \left(1- \sum_{z \in Z}p_z\right)\overline{g_m}(A \backslash Z) + \frac{1}{m} \sum_{z \in Z} g_m(m_z)$$ with $\overline{g_m}(A\backslash Z) \geq 1 - (\frac{1}{3})^{q-1}$.  
\end{itemize}
Thus we can use the same technique as that for the Shannon entropy. That is, we can track the frequency of each element $z \in Z$ separately (at most $4$ of them), and simultaneously remove all occurrences of elements in $Z$ from $A$ and apply the AMS sampling scheme to the substream $A\backslash Z$.

We only need to consider the input class $\mathcal{A}'=\{ A \in [n]^{m'}: 0 < m' \leq m, \forall i\in[n], m_i \leq 0.3 m\}$. The algorithm is essentially the same as the one for the Shannon entropy in Section~\ref{sec:Shannon}; we thus omit its description.

\begin{theorem}
\label{thm:tsallis}
There is an algorithm that maintains at the coordinator a $(1+\eps, \delta)$-approximation to $T_q(A)$ for any constant $q > 1$, using
$O \left( (k/\eps^2 + \sqrt{k}/\eps^3)\cdot \log \frac{1}{\delta}\log^4 m \right)$ bits of communication. Each site uses $O(\eps^{-2} \cdot \log\frac{1}{\delta} \log^2 m)$ and amortized $O(\eps^{-2}\cdot \log \frac{1}{\delta} \log m)$ time per item.
\end{theorem}

\begin{proof}
The algorithm and correctness proof are essentially the same as that for the Shannon entropy in Section~\ref{sec:Shannon}. We thus only analyze the costs.

Let us first bound the corresponding $\lambda_{g_m, \mathcal{A'}}, \sup_{A\in\mathcal{A'}}\pi(A)$ for $g_m$ under the input class $\mathcal{A}'$: $$E = \inf_{A'\in\mathcal{A}'} \overline{g_m}(A') = 1 - \left(\frac{1}{3}\right)^{q-1} = \Theta(1).$$ Let $h(x) = g_m(x) - g_m(x-1)$, $h'(x) < 0$. As $q > 1$,
\begin{eqnarray*}
\abs{h(m)} &=& g_m(m-1) \\
&=& m-1 - m\left(1-\frac{1}{m}\right)^q \\ 
&=& q - 1 + O(1/m).
\end{eqnarray*}
We thus set $a = q$.
On the other hand, $h(1) \leq 1$, we thus set $b = 1$. Now 
$\sup_{A\in\mathcal{A'}}\pi(A) = O\left(\frac{(1 + a/E)^2(a + b)}{a + E}\right) = O(1)$ (recall that $q$ is constant). 

Next, note that 
\begin{equation*}
\label{eq:z-1}
h(x) = -m^{1-q}( x^q - (x-1)^q) + 1 \approx 1 - q \left(\frac{x}{m}\right)^{q-1}, \text{ and}
\end{equation*} 
\begin{eqnarray*}
\label{eq:z-2}
h(x) - h((1+\eps)x) &\approx& -m^{1-q}(qx^{q-1} - q(1+\eps)^{q-1} x^{q-1}) \\
&\approx& \eps (q-1) \left(\frac{x}{m}\right)^{q-1}.
\end{eqnarray*}
Also note if $x < 0.3 m$, then $\left(\frac{x}{m}\right)^{q-1} \le \left(\frac{1}{3}\right)^{q-1}$ and $h(x) > h(0.3m) = \Omega(1)$. Therefore for $q > 1$ we can find a large enough constant $\lambda$ to make Equation~(\ref{eq:def:trackable}) hold.

For the communication cost, simply plugging $\lambda_{g_m, \mathcal{A}'} = O(1)$ and
$\sup_{A\in\mathcal{A'}}\pi(A) = O(1)$ to Equation~(\ref{eq:inf-cc}) yields our statement. Note that we have $\kappa = \Theta(\eps^{-2}\log \delta^{-1})$, hence imply the space usage and the processing time per item (using Corollary \ref{cor:ams-inf-2}).
\end{proof}

We omit the discussion on sliding windows since it is essentially the same as that in the Shannon entropy. The results are presented in Table \ref{tab:results}.





%% file: conclude.tex
\section{Concluding Remarks}
\label{sec:conclude}

In this paper we have given improved algorithms for tracking the Shannon entropy function and the Tsallis entropy function in the distributed monitoring model.  A couple of problems remain open.  First, we do not know if our upper bound is tight.  In \cite{WZ12} a lower bound of $\Omega(k/\eps^2)$ is given for the case when we have item deletions.  It is not clear if the same lower bound will hold for the insertion-only case. Second, in the sliding window case, can we keep the approximation error to be multiplicative even when the entropy is small, or do strong lower bounds exist? The third, probably most interesting, question is that whether we can apply the AMS sampling framework to track other functions with improved performance in the distributed monitoring model?  Candidate functions include Renyi entropy, $f$-divergence, mutual information, etc.  Finally,  it would be interesting to implement and test the proposed algorithms on real-world datasets, and compare them with related work competitors.

%% file: appendix.tex
\section{Proof of Lemma \ref{lem:bd}}
\label{app:lem:bd}
\begin{proof}
We can assume that $U_1,\ldots,U_m$ are distinct, since the event that $U_i = U_j$ for some $i \neq j$ has measure $0$.
Let $T_i = \min \{U_1,\ldots U_m \}$. Let $\text{OD}_i$ denote the order of $U_1,\ldots,U_i$. Let $\Sigma_i=\{ \text{Permutation of}~ [i]\}$. It is clear that for any $\sigma \in \Sigma_i$, $\Pr[\text{OD}_i = \sigma\ |\  T_i > t] = \Pr[\text{OD}_i = \sigma] = \frac{1}{i!}$.
Since the order of $U_1,\ldots,U_i$ does not depend on the minimal value in that sequence, we have
\begin{align*}
        \Pr[\text{OD}_i = \sigma,  T_i > t] 
       &=  \Pr[\text{OD}_i = \sigma\ |\  T_i > t] \cdot \Pr[T_i > t] \\
        &=  \Pr[\text{OD}_i = \sigma] \cdot \Pr[T_i > t].
\end{align*}
Therefore, the events $\{\text{OD}_i=\sigma\}$ and $\{ T_i > t\}$ are independent.

For any given $\sigma \in \Sigma_{i-1}$ and $z \in \{0,1\}$ :
\begin{eqnarray}
    &&  \Pr[J_i = z\ |\ \text{OD}_{i-1}=\sigma] \nonumber \\
      &=& \lim_{t \rightarrow 0} \Pr[J_i = z, T_{i-1} > t\ |\ \text{OD}_{i-1}=\sigma] \nonumber\\
                                &=& \lim_{t \rightarrow 0}\frac{\Pr[J_i = z\ |\ T_{i-1} > t, \text{OD}_{i-1}=\sigma] }{\Pr[\text{OD}_{i-1}=\sigma]/ \Pr[ T_{i-1} > t, \text{OD}_{i-1}=\sigma]}\label{eq:a-1}\\   
                                &=& \lim_{t \rightarrow 0}\frac{\Pr[J_i = z\ |\ T_{i-1} > t] \cdot \Pr[T_{i-1} > t] }{\Pr[\text{OD}_{i-1}=\sigma] / \Pr[\text{OD}_{i-1}=\sigma]} \label{eq:a-2}\\    
                                &=& \lim_{t \rightarrow 0} \Pr[J_i = z, T_{i-1} > t]\nonumber\\
                                &=& \Pr[J_i = z],   \nonumber
\end{eqnarray}
where (\ref{eq:a-1}) to (\ref{eq:a-2}) holds because the events $\{ J_i = z\}$ and $\{ \text{OD}_{i-1}=\sigma\}$ are conditionally independent given $\{ T_{i-1} > t\}$, and the events $\{\text{OD}_i=\sigma\}$ and $\{ T_i > t\}$ are  independent.

Therefore, $J_i$ and $\text{OD}_{i-1}$ are independent. Consequently, $J_i$ is independent of $J_1,\ldots, J_{i-1}$, since the latter sequence is fully determined by $\text{OD}_{i-1}$.

\begin{eqnarray*}
      \Pr[J_i = 1] &=& \Pr[U_i < \min \{U_1, \ldots, U_{i-1} \}] \\
                 &=& \int_0^1(1-x)^{i-1}dx 
                 = \frac{1}{i}.
\end{eqnarray*}
    Thus $\E[J_i] = \Pr[J_i = 1] = \frac{1}{i}$. By the linearity of expectation, $\E[J] = \sum_{i \in [m]} \frac{1}{i} \approx \log m$.

    Since $J_1, \ldots, J_m$  are independent, $\Pr[J > 2 \log m] < m^{-1/3}$ follows from a Chernoff Bound.
\end{proof}

\section{The Assumption of Tracking $m$ Exactly}
\label{app:why:m}
We explain here why it suffices to assume that $m$ can be maintained at the coordinator {\em exactly without any cost}.  First, note that we can always use \CEN\ to maintain a $(1+\eps^2)$-approximation of $m$ using $O(\frac{k}{\eps^2}\log m)$ bits of communication, which will be dominated by the cost of other parts of the algorithm for tracking the Shannon entropy. 
Second, the additional error introduced for the Shannon entropy by the $\eps^2 m$ additive error of $m$ is negligible: let $g_x(m) = f_m(x) - f_m(x-1) = x\log\frac{m}{x} - (x-1)\log\frac{m}{x-1}$, and recall (in the proof of Lemma \ref{lem:trackable}) that $X > 0.5$ under any $A\in\mathcal{A'}$. It is easy to verify that
$$
	\left|g_x((1\pm \eps^2)m) - g_x(m)\right| = O(\eps^2)  \leq O(\eps^2) X,
$$
which is negligible compared with $|X - \hat{X}| \leq O(\eps) X$ (the error introduced by using $\hat{R}$ to approximate $R$). Similar arguments also apply to $X'$, and to the analysis of the Tsallis Entropy.

\section{Pseudocode for \CEN}
\label{app:CEN}
Algorithm \ref{algo:ReceiveItemCEN}, \ref{algo:SampleUpdateCEN} describe how we can maintain a $(1 + \eps)$-approximation to the frequency of element $e$.
\medskip

\begin{algorithm}[H]
initialize $c \gets 1, ct \gets 0$\;
\ForEach{ $e$ received}{
  $ct \gets ct + 1$\;
  \If{$ct = 1$}{
    send a bit to the coordinator\;
  }
  \If {$ct > (1 + \eps) c$}{
        $c \gets ct$\;
        send a bit to the coordinator\;
    }
}
\caption{ Receive an item $e$ at a site}
\label{algo:ReceiveItemCEN}
\end{algorithm}

\begin{algorithm}[H]
initialize $ct_i \gets 0$ for all $i \in [k]$\;
initialize $c \gets 0$\;
\tcc*[h]{maintain $c$ as the count}

\While{True}{
  \If{received a bit from site $i$}{
    \If{$ct_i = 0$}{
      $c \gets c + 1$\;
      $ct_i \gets 1$\;
    } \Else{
      $c \gets c + \eps\cdot ct_i$\;
      $ct_i \gets (1 + \eps) ct_i$\; 
    }
  }  
}
\caption{{\bf \CEN}$(e, \eps)$ maintains $c$ as the count at the coordinator}
\label{algo:SampleUpdateCEN}
\end{algorithm}